\documentclass[a4paper,UKenglish,cleveref, autoref]{lipics-v2019}

\title{\mbox{Controlling a Random Population is \EXPTIME-hard}}

\titlerunning{}

\author{Corto Mascle}{ENS Paris-Saclay, FR} {corto.mascle@ens-paris-saclay.fr}{}{}
\author{Mahsa Shirmohammadi}{CNRS, IRIF, Universit\'e de Paris, FR}{mahsa@irif.fr}{}{}
\author{Patrick Totzke}{University of Liverpool, UK}
{totzke@liverpool.ac.uk}{http://orcid.org/0000-0001-5274-8190}{}

\authorrunning{C.~Mascle, M.~Shirmohammadi, P.~Totzke}
\Copyright{P.~Totzke}


\category{}

\relatedversion{}

\supplement{}

\nolinenumbers 

\hideLIPIcs  

\EventEditors{John Q. Open and Joan R. Access}
\EventNoEds{2}
\EventLongTitle{42nd Conference on Very Important Topics (CVIT 2016)}
\EventShortTitle{CVIT 2016}
\EventAcronym{CVIT}
\EventYear{2016}
\EventDate{December 24--27, 2016}
\EventLocation{Little Whinging, United Kingdom}
\EventLogo{}
\SeriesVolume{42}
\ArticleNo{23}
\usepackage{macros}
\usepackage{tikz}
\usepackage{tikzmacros}
\usepackage{todonotes}

\newcommand{\mypar}[1]{\vspace{-1em}\subparagraph*{#1}}
\begin{document}
\maketitle
\begin{abstract}
    Bertrand et al.~\cite{BDGGG2019} describe two-player zero-sum games in which one player tries to achieve a reachability objective in $n$ games (on the same finite arena) simultaneously by broadcasting actions, and where the opponent has full control of resolving non-deterministic choices. They show \EXPTIME\ completeness for the question if such games can be won for every number $n$ of games.

We consider the \emph{almost-sure} variant in which the opponent randomizes their actions, and where the player tries to achieve the reachability objective eventually with probability one.
The lower bound construction in \cite{BDGGG2019} does not directly carry over to this randomized setting.\footnote{The construction in Theorem~6.1 would allow spurious wins for the controller by alternating actions \texttt{init} and \texttt{reset} until an incorrect initialization of the ATM is reached. This must happen eventually with probability one and can subsequently be exploited by controller to win directly.}
In this note we show \EXPTIME\ hardness for the almost-sure problem by reduction from
Countdown Games.
\end{abstract}

\section{Definitions}
\subparagraph*{Population Control.} 
Write 
$\?M=(Q,\act,\prob)$
for a Markov Decision Process,
where
$Q$ and $\act$, are finite sets of \emph{states}
and \emph{actions},
and $\prob:Q\x\act \to \dist(Q)$ assigns to each state and action
a probability distribution over states.
A \emph{successor} of~$q$ on action~$a$ is a state $p$ with $\delta(q,a)(p)>0$.
The $n$-fold synchronized product of $\?M$ is the MDP 
$\?M^{(n)}=(Q^{(n)},\act,\prob)$
whose states, called \emph{configurations} here, 
are $n$-dimensional vectors with components in $Q$,
and $\delta$ is lifted to $Q^{(n)}$ in the natural way:
$\prob(\vec{q},a)(\vec{p})=\prod_{i=0}^{n-1}\prob(\vec{q}(i),a)(\vec{p}(i))$
for all $\vec{q},\vec{p}\in Q^{(n)}$ and $a\in\Sigma$.
A \emph{strategy} $\sigma:Q^{(n)}\to \act$ assigns to every configuration~$\vec{q}$ an action\footnote{
We consider here only memoryless deterministic strategies
as those suffice for almost-sure reachability problems \cite{Puterman:book}.
}.
For every initial state $\vec{q}\in Q^{(n)}$, such a strategy
induces a probability space $(\Omega,\Prob{\vec{q}}^\sigma)$ over all infinite sequences in~$\Omega=\vec{q}(Q^{(n)})^\omega$
(see \cite{Puterman:book} for details).

In a configuration $\vec{q}\in Q^{(n)}$, we say $m$ components \emph{mark} state $p\in Q$ if the number of different indices $0\le i <n$ with $\vec{q}(i) = p$ is equal to~$m$. In case $m>0$ we simply call the state $p$ \emph{marked} in~$\vec{q}$.
Let $\mathbf{{Start}}, \mathbf{End}\in Q^{(n)}$ denote the configurations
in which all $n$ components mark a designated initial (or final, resp.) state of $\?M$. 
A configuration $\vec{q}\in Q^{(n)}$ 
\emph{can be synchronized} if there exists a strategy $\sigma$ such that the probability $\Prob{\vec{q}}^{\sigma}((Q^{(n)})^*\textbf{End})$, of eventually visiting $\mathbf{End}$, is one. 
$\?M^{(n)}$ \emph{can be synchronized} if~$\textbf{Start}$ can be synchronized.

Given an MDP~$\?M$ equipped with  initial and  final states, the
\textsc{PopulationControl} problem asks whether $\?M^{(n)}$ can be synchronized for every $n\in\+N$.
\mypar{Countdown Games.}
A \emph{Countdown Game}
is given by a directed graph $\?G=(V,E)$,
where edges carry positive integer weights, 
$E\subseteq (V\x\+N_{>0}\x V)$. 
For an initial pair $(v,c_0)\in V\x\+N$ of a vertex and a number, 
two opposing players (Player~1 and~2)
alternatingly determine a sequence of such pairs as follows.
In each round, from $(v,c)$,
Player~1 picks a number $d\le c$ such that $E$ contains at least one edge $(v,d,v')$;
then Player~2 picks one such edge and the game continues from $(v',c-d)$.
Player~1 wins the game iff the play reaches a pair in $V\x\{0\}$.

\textsc{CountdownGame} is the decision problem which asks if Player~1 has a strategy to win a given game for a given initial pair $(v_0,c_0)$. 
All constants in the input are written in binary.
\color{black}

\begin{proposition}[Thm.~4.5 in \cite{JLS2007}]
    \label{lemma}
    \textsc{CountdownGame} is \EXPTIME-complete.
\end{proposition}

\newcommand{\HEAVEN}{\texttt{Heaven}}
\newcommand{\HELL}{\texttt{Hell}}
\newcommand{\DIV}[2]{{#1}.{#2}}

\section{The Reduction}
\label{sec:gadgets}
In order to reduce \textsc{CountdownGame} to \textsc{PopulationControl} we first 
observe that the number of turns in a Countdown Game cannot exceed the initial value of the counter, as the initial counter value decreases at each turn. Thus, if Player~2 has a winning strategy, choosing actions at random yields a positive probability of applying that strategy, hence a positive probability of winning. Therefore Player~1 wins the initial game if, and only if, she wins with probability one against a randomized adversary.

The main idea for our further construction 
is to require Player~1 to move components one-by-one 
away from a waiting state, first into 
the control graph of the Countdown Game, and ultimately into the goal.
To avoid a loss in the intermediate phase she needs to win an instance of that game against a randomizing opponent.
This is enforced using a combination of gadgets, including two
binary counters that 
can effectively test for zero, be set to specific numbers,
and that are set up so that they can decrement at the same rate.
As a result, Player~1 has a winning strategy for the two-player Countdown Game if, and only if, the controller
can synchronize the $n$-fold product of the constructed MDP
for all $n$.

\medskip
For a given Countdown Game $\?G$ with an initial pair~$(v_0,c_0)$ we construct an MDP $\?M$ as follows.
We write that action $a$ \emph{takes state} $s$ to \emph{successor} $t$ to mean that $\delta(s,a)(t)>0$.
The exact probability distributions $\delta$ do not matter in our construction so we let $\delta(s,a)$ be the uniform distribution over such successors.

Whenever action $a$ takes state $s$ only back to itself we say that $s$ \emph{ignores} $a$.
There are states $\HEAVEN$ (the target) and $\HELL$ which ignore all actions.
For a given state $s$, an action $a$ is \emph{angelic} if it takes $s$ only to $\HEAVEN$,
and \emph{daemonic} if it takes $s$ to $\HELL$.
An action $a$ is \emph{safe} in a configuration if it is not daemonic for any marked state (in any gadget).

\begin{figure}[t]
\begin{subfigure}[T]{0.4\textwidth}
\begin{tikzpicture}[
  node distance=0.75cm and 1.25cm,
  ]
\node[estate] (I)                    {\texttt{wait}};
\node[estate]         (R)  [right= of I]      {\texttt{ready}};

\path  (I.north east) edge node[anode,auto] {$\texttt{wait}$}(R.north west);
\path  (R.south west) edge node[anode,auto] {$\texttt{wait}$}(I.south east);

\path  (I) edge[loop left] node[anode]{$\texttt{wait}$} (I);
\end{tikzpicture}

\end{subfigure}
\begin{subfigure}[T]{0.6\textwidth}
\hfill
\begin{tikzpicture}[
  node distance=0.75cm and 1.5cm,
  ]

\node[estate] (1)  {$W$};
\node[estate,right= of 1] (2)  {$G$};
\node[estate,right= of 2] (3)  {$A$};
\node[estate,right= of 3] (4)  {$B$};

\path (1) edge[loop left] node[anode,auto] {$\texttt{wait}$}(1);
\path (1.north east) edge node[anode,auto] {$\texttt{go}$} (2.north west);
\path (2.north east) edge node[anode,auto] {$\Sigma_{\?G}$} (3.north west);
\path (3.south west) edge node[anode,auto] {$\texttt{next}$} (2.south east);
\path (4.south west) edge node[anode,below] {$\Sigma_{AC}$}(3.south east);
\path (3.north east) edge node[anode,above] {$\Sigma_{MC}$}(4.north west);
\path (2.south west) edge node[anode,auto] {$\texttt{win}$} (1.south east);

\end{tikzpicture}

\end{subfigure}
\caption{The waiting (on left) and
  the control gadgets (on right).
Edges labelled by $\Sigma_X$ are shorthand
    for several edges, one for each action in $\Sigma_X$.
    All but the depicted actions are daemonic.
}
\label{fig:control}
\label{fig:waiting-room}
\end{figure}
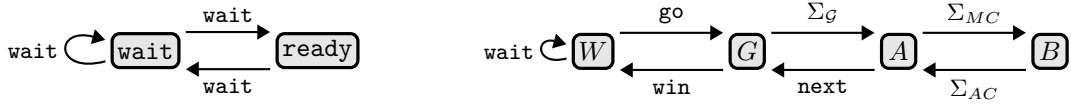
Besides the special states $\HEAVEN$ and $\HELL$, $\?M$ contains several gadgets described below.
\vspace{-1em}
\mypar{Waiting.}
\newcommand{\WaitAct}{\texttt{wait}}
\newcommand{\WaitState}{\texttt{Wait}}
\newcommand{\ReadyState}{\texttt{Ready}}
The waiting gadget
has two states $\WaitState$ and $\ReadyState$
which react to the action $\texttt{wait}$
as depicted in Figure~\ref{fig:waiting-room} (left).
Whenever a configuration marks one of these states,
a strategy that continuously plays $\WaitAct$ will almost-surely reach a configuration
in which exactly one component marks $\ReadyState$.

A special action $\texttt{go}$ (to indicate successful isolation of one component) takes
$\ReadyState$ to the initial state $v_0$ of the game $\?G$.
All other actions (in gadgets described below) are ignored.

\mypar{Game.}
The game $\?G=(G,E)$ is directly interpreted as MDP: 
For every edge $(s,d,s')\in E$ there is an action $(s,d)$ 
which takes $s$ to $s'$ 
and which is daemonic for all states $s'\neq s$.

The action $\texttt{win}$ is is angelic for every state of $G$.  All other actions are ignored.

\mypar{Binary Counters.}
A ($k$-bit) Counter consists of states $\BIT{i}{j}$ for all $0\le i< k$ and $j\in\{0,1\}$.
For every bit $i$ there is a decrement action $\dec{i}$ which
\begin{itemize}
    \item takes $\BIT{j}{0}$ only to $\BIT{j}{1}$ for all $0\le j < i$,
    \item takes $\BIT{i}{1}$ only to $\BIT{i}{0}$,
    \item is daemonic for $\BIT{i}{0}$, and
    \item is ignored by all $\BIT{j}{l}$, for all $i<j$ and $l\in\{0,1\}$.
\end{itemize}

We say that a configuration \emph{holds} the number 
$c<2^k$ in this counter if 
it marks those states that represent the binary expansion of $c$:
for all~$0\le i\le k-1$,
state $\BIT{i}{j}$ is marked iff
the $i$th bit in the binary expansion of $c$ is~$j$. 
An action~$a$ \emph{sets} the counter to number $d$
if for all $0\leq i<k$, 
it takes $\BIT{i}{0}$ to only $\BIT{i}{j}$ where $j\in \{0,1\}$ is the $i$th bit in the binary expansion of~$d$,
and is daemonic for all $\BIT{i}{1}$ (to ensure that the counter can only be set if it holds~$0$).
Observe that if a counter holds $c$ then there is a unique maximal sequence of safe decrement actions,
that has length $c$ and after which the counter holds $0$.

Additionally, for every bit $i$ the gadget has an error action $\error{i}$,
which is daemonic for $\BIT{i}{0}$ and $\BIT{i}{1}$, and 
angelic for every other state (of $\?M$).
These actions can be used to quickly synchronize any configuration in which the counter is not correctly initialized,
i.e., does not hold a number.
See Figure~\ref{fig:counter} for a depiction of a $4$-bit counter.

\medskip
The MDP $\?M$ will contain two distinct counter gadgets. 
A main counter $MC$ 
has 
$\log_2(n_0)$ bits to hold possible counter values of the Countdown Game.
An auxiliary counter $AC$ has $\log_2(d_{\texttt{max}})$ many bits to hold the largest edge weight~$d_{\texttt{max}}$ in $\?G$.
These have distinct sets of states and actions,
so for clarity, 
we write $\DIV{C}{x}$ to refer to state (or action) $x$ in gadget $C$.
We connect some new actions to these two counters as follows.
\begin{itemize}
    \item The action $\texttt{go}$ sets $MC$ to~$n_0$;
this ensures that $MC$ holds $n_0$ when starting to simulate~$\?G$.
\item The action $\texttt{win}$ is daemonic for every state $\DIV{MC}{\BIT{i}{1}}$.
 This enforces that the $MC$ must hold $0$ when a strategy claims Player~1 wins $\?G$.
\item Any action~$(v,d)\in \Sigma_{\?G}$ sets $AC$ to $d$;
\item The action $\texttt{next}$ is daemonic for every state $\DIV{AC}{\BIT{i}{1}}$.
This enforces that a strategy must first count down from~$d$ to~$0$ before it can simulate the next move in~$\?G$.
\end{itemize}
\begin{figure}[t]
    \centering
\begin{tikzpicture}[
  node distance=1cm and 2.75cm,
  ]

\node[estate] at (0,0) (Z0)  {$\Bit{0}{0}$};
\node[left= of Z0, estate] (Z1)  {$\Bit{1}{0}$};
\node[left= of Z1, estate] (Z2)  {$\Bit{2}{0}$};
\node[left= of Z2, estate] (Z3)  {$\Bit{3}{0}$};

\node[above=  of Z0, estate] (O0)  {$\Bit{0}{1}$};
\node[left= of O0, estate] (O1)  {$\Bit{1}{1}$};
\node[left= of O1, estate] (O2)  {$\Bit{2}{1}$};
\node[left= of O2, estate] (O3)  {$\Bit{3}{1}$};

\path  (Z0.north west) edge node[anode,align=left,auto] {$\Dec{1}$,\\$\Dec{2}$,\\$\Dec{3}$ }(O0.south west);
\path  (O0.south east) edge node[anode,right] {$\Dec{0}$}         (Z0.north east);
\path  (Z1.north west) edge node[anode,align=left,left] {$\Dec{2}$,\\$\Dec{3}$}(O1.south west);
\path  (O1.south east) edge node[anode,auto] {$\Dec{1}$}         (Z1.north east);
\path  (Z1) edge[loop left] node[anode,auto] {$\Dec{0}$}(Z1);
\path  (O1) edge[loop left] node[anode,auto] {$\Dec{0}$}(O1);

\path  (O2.south east) edge node[anode,auto] {$\Dec{2}$}         (Z2.north east);
\path  (Z2) edge[loop left] node[anode,auto,align=left] {$\Dec{0}$,\\$\Dec{1}$}(Z2);
\path  (O2) edge[loop left] node[anode,auto,align=left] {$\Dec{0}$,\\$\Dec{1}$}(O2);
\path  (Z2.north west) edge node[anode,align=left,left] {$\Dec{3}$ }(O2.south west);

\path  (O3.south east) edge node[anode,auto] {$\Dec{3}$}         (Z3.north east);
\path  (Z3) edge[loop left] node[anode,auto,align=left] {$\Dec{0}$,\\$\Dec{1}$,\\$\Dec{2}$}(Z3);
\path  (O3) edge[loop left] node[anode,auto,align=left] {$\Dec{0}$,\\$\Dec{1}$,\\$\Dec{2}$}(O3);
\end{tikzpicture}

\caption{A (4-bit) Binary Counter.
    Not displayed are edges labelled by $\dec{i}$ that make the respective
    actions daemonic for state $\BIT{i}{0}$, and error actions $\error{i}$,
    which are daemonic for $\BIT{i}{0}$ and $\BIT{i}{1}$, for all bits
    $i\in\{0,1,2,3\}$.
}
\label{fig:counter}
\end{figure}
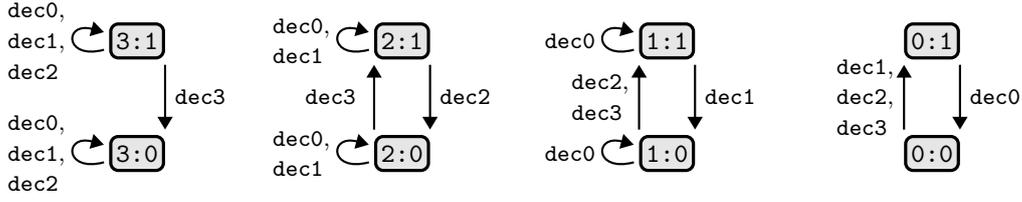
\mypar{Control.}
The control gadget
will enforce that a synchronizing strategy proposes actions in a proper order; see Figure~\ref{fig:control}.
It consists of states $W,G,A,B$, and contains actions of all gadgets above (including $\texttt{go}$,  $\texttt{win}$, $\texttt{next})$
and a new $\error{}$ action, which is angelic for all states except~$W$, for which it is daemonic.
All omitted edges in~Figure~\ref{fig:control} are daemonic.

\mypar{Start/End.}
To complete the construction of~$\?M$, we introduce 
an initial state $\StartState$ and
actions~$\texttt{start}$ and $\texttt{end}$.
The action $\texttt{start}$ 
takes $\StartState$ to $\WaitState$ (Waiting gadget), $W$    
(Control gadget), and all $\BIT{i}{0}$ states of counters $AC$ and $MC$.
It is daemonic for every other state.

The  action $\texttt{end}$ is daemonic for $\WaitState$ and $\ReadyState$,
and angelic for every other state in~$\?M$.

\begin{lemma}
    \label{theorem}
    $\?M^{(n)}$ is synchronizable for all $n\in\+N$ iff Player~1 wins $\?G$.
\end{lemma}
\begin{proof}

Suppose Player~1 wins the game~$\?G$.
Fix $n$. Recall that in $\?M^{(n)}$  all components of the initial configuration  mark~$\StartState$. A synchronizing strategy proceeds as follows:
\begin{itemize}
        \item Play $\texttt{start}$ to initialize the Waiting and Control gadgets, and to set~$AC$ and $MC$ to~$0$.
        If any of the gadgets is not correctly initialized afterwards, play the respective error action to win directly.
           For instance, if $W$ is unmarked, play $\error{}$ to synchronize.
        \item 
            Reduce the number of components marking $\WaitState$ one by one 
            until a configuration is reached in which $\WaitState$ is not marked.
            Once this is true, play~$\texttt{end}$ to synchronize.
        \item To reduce the number of components marking $\WaitState$,
            isolate one of them, and move it to $\HEAVEN$ by simulating the Countdown Game:
            \begin{enumerate}
            \item Play $\WaitAct$ until only a single component marks $\ReadyState$, then play $\texttt{go}$.
                This will mark~$v_0$ in the game gadget and sets~$MC$ to $n_0$. Recall that $(v_0,n_0)$ is the initial pair of~$\?G$.
            \item Simulate rounds of the game~$\?G$:
           assume state~$v$ in the game gadget is marked and
           the counter~$MC$ holds~$c$, then let $d$ be
           the the number Player~$1$ plays to win
           from the pair~$(v,c)$ in~$\?G$. Play~$(v,d)$.
           This action will  set  $AC$ to $d$.
           Alternate between (safe) decrement actions in $AC$ and $AB$ until they hold $0$ and $c-d$, respectively.
           Play $\texttt{next}$.
       \item The above simulation of rounds in~$\?G$ is repeated until both $AC$ and $AB$ hold $0$, by assumption that Player~1 wins~$\?G$ this is possible. At this point it is safe to play $\texttt{win}$.
            \end{enumerate}
    \end{itemize}

   Conversely, assume that Player~1 cannot win~$\?G$.
    Suppose that after the (only possible) initial move $\texttt{start}$, all gadgets are correctly initialized.
    Clearly, for every $n$, this event has strictly positive probability. We argue that no strategy can synchronize
    such a configuration. Indeed, a successful strategy had to play a sequence in
    $\WaitAct^*\cdot \texttt{go}$ first, 
    followed by actions in $(\Sigma_{\?G}\cdot(\Sigma_{AC}\cdot\Sigma_{MC}\cdot\texttt{next})^*)^*$,
    by construction of the control gadget.
    If after playing $\texttt{go}$, more than one component mark~$v_0$, there is a non-zero chance that
    these will diverge, making subsequent actions in $\Sigma_{\?G}$ unsafe.
    If exactly one component marks~$v_0$ then the second sequence of actions (assuming all actions are safe) corresponds to a  play of~$\?G$.
    This inevitably leads to a configuration in which
    counter $MC$ holds $0$ and the control enforces that the next action is in $\Sigma_{MC}$.
    But any such action will be daemonic for some state in $MC$ and thus not be safe.
    We conclude that every strategy will lead to a configuration that  at least one component marks $\HELL$
    and thus cannot be synchronized.    
\end{proof}

Our claim follows immediately from \cref{lemma,theorem}.
\begin{theorem}
    \textsc{PopulationControl} is \EXPTIME-hard.
\end{theorem}

\vfill
\bibliography{journals,conferences,autocleaned}
\end{document}